\documentclass[aps,pra,twocolumn,floatfix,
nofootinbib,showpacs,longbibliography]{revtex4-1}
\usepackage[british]{babel}  
\usepackage[sc,osf]{mathpazo}
\usepackage[scaled=0.86]{berasans}  
\usepackage[colorlinks=true, citecolor=blue, urlcolor=blue]{hyperref}  
\usepackage{graphicx} 
\usepackage[babel]{microtype}  
\usepackage{amsmath,amssymb,amsthm,bm,amsfonts,mathrsfs,bbm} 

\usepackage{xspace}  
\usepackage{pgf,tikz}
\usepackage{xcolor}
\usepackage{multirow}
\usepackage{array}
\usepackage{bigstrut}
\usepackage{braket}
\usepackage{color}
\usepackage{natbib}
\usepackage{amsmath}
\usepackage{multirow}
\usepackage{mathtools}
\usepackage{float}
\usepackage[caption = false]{subfig}
\usepackage{xcolor,colortbl}
\usepackage{color}
\newcommand{\Tr}{\operatorname{Tr}}

\newcommand{\be}{\begin{equation}}
\newcommand{\ee}{\end{equation}}
\newcommand{\ba}{\begin{eqnarray}}
\newcommand{\ea}{\end{eqnarray}}
\newcommand{\ketbra}[2]{|#1\rangle \langle #2|}

\newtheorem{theorem}{Theorem}

\begin{document}

\title{Thermodynamic advancement in the causally inseparable occurrence of thermal maps}

\author{Tamal Guha}
\email{g.tamal91@gmail.com}    
\affiliation{Physics and Applied Mathematics Unit, Indian Statistical Institute, 203 B.T. Road, Kolkata 700108, India.}

\author{Mir Alimuddin}
\email{aliphy80@gmail.com}    
\affiliation{Physics and Applied Mathematics Unit, Indian Statistical Institute, 203 B.T. Road, Kolkata 700108, India.}

\author{Preeti Parashar}
\email{parashar@isical.ac.in}    
\affiliation{Physics and Applied Mathematics Unit, 
	Indian Statistical Institute, 
	203 B.T. Road, Kolkata 700108, India.}

\begin{abstract}
Quantum mechanics allows the occurrence of events {\it without} having any definite causal order. Here, it is shown that the application of two different thermal channels in the causally inseparable order can enhance the potential to extract work, in contrast to any of their definite (separable) order of compositions. This enhancement is also possible even {\it without} assigning any {\it thermodynamic resource} value to the controlling qubit. Further, we provide the first non-trivial example of causal enhancement with {\it non-unital pin} maps, for which it is still {\it not clear} how to obtain a superposition of path structure (under definite causal order). Hence, it may be a potential candidate to accentuate the difference between {\it superposition of time} and {\it superposition of path}.
\end{abstract}



\maketitle
\section{Introduction}
The second law of thermodynamics, one of the most fundamental principles of science, directs the occurrence of any physical process. In addition, the execution of such a process actually makes an irreversible transformation of the initial state of the system to its final form. For a thermodynamic system, any process respecting the second law, finally leads to an equilibrium state for which the population in different energy levels follows the Gibbs-Boltzmann distribution. This is the ultimate state for any arbitrary thermodynamic system evolving under the thermalization process as the irreversiblity of energy to work conversion achieves its maximum limit for this distribution. To quantify this irreversibility, a thermodynamic potential, namely {\it "free energy"} is defined as $F=U-TS$, where $U,T,S$ stand for the internal energy, temperature and the entropy of the system respectively, which takes the minimum value for the above mentioned distribution.\par 
In the quantum domain, the state which exhibits this distribution on energy eigen-states are termed as $\beta${\it -thermal} state, represented as $\frac{e^{-\beta \hat{H}}}{Tr(e^{-\beta \hat{H}})}$, where $\hat{H}$ is the governing Hamiltonian of the system\cite{lenard}. On the other hand, the most general quantum operation can be visualized as a completely positive trace preserving (CPTP) map acting on a quantum state. Hence, any thermalization process can be represented by a CPTP map and thus can be termed as a thermal channel. For the finite particle scenario, under the action of this CPTP map, not only the free energy (as defined earlier) corresponding to the initial quantum state, but all the Renyi-$\alpha$ free energy $F_{\alpha}, \alpha\in [0,\infty)$ will decrease under the action of these thermal channels \cite{brandao'PNAS}, where $F_{\alpha}(\rho)=Tr(\rho H)-T S_{\alpha}(\rho)$ and $S_{\alpha}(\rho)$ is the Renyi-$\alpha$ entropy. These $\alpha-$free energies play a significant role in order to determine the direction of the state transformation in the microscopic regime, even with the assistance of an additional catalyst system.  However, in the asymptotic limit, the well known {\it free energy} (for $\alpha\to1$) is the only determining parameter. 
On the other hand, a significant connection between memory capacity of thermal channels and the free energy of the same has recently been established in \cite{Varun'PRL}.\par 
However, all these studies have one common feature which is to understand the departure (or, validity) of traditional thermodynamics  in the quantum regime. Alternatively, the superiority of non-classical features, viz. entanglement and discord of quantum correlations have been studied from the perspective of work extraction in \cite{huber'PRX,our'PRA,demonic_ergo}. The key feature in all these is the {\it quantum superposition}, presence of which at the single particle level can cause locking of work extraction in a different scenario\cite{popescu'bath}. In this paper, we have encountered the superiority in work extraction through the presence of quantum superposition in an assisting qubit.\par 
Here, we have considered the action of two different thermal channels on a quantum system. It is obvious from the above discussion that for the final state, free energy should be non-increasing, hence the extractable amount of work serves as a monotone under these channel actions. However, it is interesting to ask whether there can be an increment in the extractable work, in the presence of another assisting quantum system. More precisely, one can apply initially an energy preserving unitary on the system-ancilla joint state to build a correlation among them and then allow the system qubit to evolve through the above mentioned channels. Thereafter, communicating the outcomes of a preferred measurement on the assisting quantum state, one can enhance the amount of extractable work form the resulting system qubit. On the other hand, one can directly use this ancillary system to control the order of these thermal channels coherently, which we will refer to as {\it causally inseparable occurrence of thermal maps} (CIOTM). Exhibiting the power of superposition as the controller of quantum processes have drawn lot of attraction recently from the perspective of several information processing tasks. Firstly, in \cite{brukner'Nat}, the authors have introduced an inequality on the average pay-off for a {\it GYNI}-type game, violation of which ascertains the presence of a nontrivial correlation shared between the players, which can not have any pre-defined causal structure. Thereafter, it is also shown that, the indefinite causal order between the action of two zero classical capacity channels can transfer non-zero information \cite{Ebler'PRL}, and more surprisingly, causally inseparable occurrence of two entanglement breaking channels can establish perfect entanglement between the sender and the receiver\cite{our'arxiv}. However, in \cite{abott}, Abbott {\it et. al.} have exhibited an advancement by controlling coherently the path of system qubit in classical as well as in quantum information processing without invoking any causal indefiniteness. Not withstanding the criticisms regarding the {\it (non)unique} indefinite causal advancements \cite{brukner'PRA}, in the present work we have introduced a CIOTM scenario and exhibited a surprising feature what makes it unique in comparison to other existing results ({\it depicted in Fig \ref{proc}}).\par 
This paper is organized as follows: In Section II  preliminary concepts regarding quantum channels have been discussed, with a special focus on relevant thermal channels. Section IIIA mainly focuses on the enhancement of extractable work under CIOTM scenario with expense of a resourceful controller qubit, whereas, Section IIIB deals with the thermodynamically free controller qubit. Finally, we have concluded our results in Section IV. A more general study on the parameters of the thermal channels is discussed in the Appendix.
\section{Preliminaries}
In general, a quantum operation $\mathcal{N}$ acting on a state $\rho\in\mathcal{D}(\mathcal{H})\subset\mathcal{L}(\mathcal{H})$ is a completely positive trace preserving (CPTP) map, where, $\mathcal{L}(\mathcal{H})$ and $\mathcal{D}(\mathcal{H})$ represents respectively the set of all linear bounded operators and density matrices over the Hilbert space $\mathcal{H}$. The action of any CPTP map on $\rho$ can be represented as $\mathcal{N}(\rho)=\sum_{k}E_{k}\rho E_{k}^{\dagger}$, where $E_{k}$'s are the Kraus operators acting on  $\mathcal{H}$, which satisfy $\sum_{k}E_{k}^{\dagger}E_{k}=\mathbb{I}$. However, the action of a quantum channel on an input quantum state can be visualized as the marginal dynamics of a global unitary $\mathbb{U}$ acting jointly on $\rho\otimes\ketbra{e_{0}}{e_{0}}$, where $\rho$ is the system state and $\ket{e_{0}}$ is the initial environment state, i.e., 
\begin{equation}\label{e1}
	\mathcal{N}(\rho)= Tr_{E}[\mathbb{U}(\rho\otimes\ketbra{e_{0}}{e_{0}})\mathbb{U}^{\dagger}]=\sum_{k}E_{k}\rho E_{k}^{\dagger}
\end{equation}
where, $E_{k}=\bra{e_{k}}\mathbb{U}\ket{e_{0}}$ are the above mentioned {\it Kraus operators} acting on $\mathcal{D}(\mathcal{H})$.\par 
A CPTP map acts as a thermal channel on the system qubit, if $\mathbb{U}$ is an energy-preserving unitary on the global system and a thermal state $\tau_{\beta}$ is taken as the environment state in Eq.\eqref{e1}. Under the action of these channels, the entropic distance, and hence the free energy difference between system and the thermal state is strictly non-increasing. Further, this distance determines the rate of state convertibility \cite{brandao'PRL}. However, to assign a CPTP map as a thermal channel, all the Renyi-$\alpha$ free energies should monotonically decrease under the channel action \cite{brandao'PNAS}. On the other hand, the extractable amount of work from a single particle quantum system (single-shot) is again in disagreement with the traditional free energy difference, which has been exhibited in \cite{Horodecki'Nature}. There the authors have enunciated two different free energies, to quantify the amount of extractable work in presence of a thermal bath and the work cost in formation of a particular state in the quantum nanoscale regime. Though the complications regarding the characterization of all the free energies makes it difficult to identify the class of thermal channels in general, it is easy to argue that any pin map which maps the total quantum state space to a fixed thermal state is always a thermal channel. Further, it also follows from \cite{Horodecki'Nature,Varun'PRL} that the domain of a quantum state under thermal process is convex and hence the actions itself form a convex set.\par 
From the above argument it is clear that the {\it generalized amplitude damping} channel $(\mathcal{N}_{GAD})$ with unit annihilation parameter and {\it phase flip} channels $(\mathcal{N}_{PF})$ are valid thermal maps. The Kraus operators corresponding to these channels are $\{E_{i}\}_{i=0}^3$ and $\{F_{j}\}_{j=0}^1$, where
\begin{center}
	$E_0=\sqrt{p}\begin{bmatrix}
	1&0\\
	0&\sqrt{1-\gamma}\\
	\end{bmatrix}$,
	$E_1=\sqrt{p}\begin{bmatrix}
	0&\sqrt{\gamma}\\
	0&0\\
	\end{bmatrix}$,\\
	$E_2=\sqrt{1-p}\begin{bmatrix}
	\sqrt{1-\gamma}&0\\
	0&1\\
	\end{bmatrix}$,	
	$E_3=\sqrt{1-p}\begin{bmatrix}
	0&0\\
	\sqrt{\gamma}&0\\
	\end{bmatrix}$,
\end{center}
and 
\begin{center}
	$F_0=\sqrt{q}\begin{bmatrix}
	1&0\\
	0&1\\
	\end{bmatrix}$,
	$F_1=\sqrt{1-q}\begin{bmatrix}
	1&0\\
	0&-1\\
	\end{bmatrix}$.
\end{center}
Here $\gamma$ is the {\it annihilation parameter} and $p\in[0,1]$ is the parameter for equilibrium state corresponding to the actions of $E_{k}$'s, whereas $q\in[0,1]$ denotes the probability to keep the initial state unaltered through the actions of $F_{k}$'s \cite{nc_book}. It is easy to verify that for $\gamma=1$ the channel $\mathcal{N}_{GAD}$ becomes a pin map, which maps the entire Bloch sphere of qubit state-space to a single point $\tau_{p}=p\ketbra{0}{0}+(1-p)\ketbra{1}{1}$, which will serve the purpose of a thermal map for $\frac{1}{2}\leq p\leq 1$. On the other hand, the action of $\mathcal{N}_{PF}$ on the Bloch sphere can be visualized as a shrunk ellipsoid of equatorial radius $(2q-1)$ and keeping the $Z$-axis unaltered.
\section{Main Results}
We will first show that no {\it separable} order for the occurrence of the generalized amplitude damping and phase-flip channels can increase the free energy of the initial state. This motivates us to design a protocol by controlling the actions coherently to obtain higher free energetic states.\par
\begin{theorem}
	No incoherent (or, causally separable) combination of the channels $\mathcal{N}_{GAD}$ and $\mathcal{N}_{PF}$ can enhance the free energy of the initial quantum state.
\end{theorem}
\begin{proof}
If the compositions of these channels is such that $\mathcal{N}_{PF}$ is followed by $\mathcal{N}_{GAD}$, then this will trivially map the initial state to $\tau_{p}$ as mentioned earlier. However, it is interesting to observe the reverse composition  $\mathcal{N}_{PF}\circ\mathcal{N}_{GAD}$. The final map here is the phase-flip channel, which preserves any qubit state diagonal in energy  ($\sigma_{Z}$) Eigen basis. As the action of $\mathcal{N}_{GAD}$ produces $\tau_{p}=p\ketbra{0}{0}+(1-p)\ketbra{1}{1}$ is diagonal in the same basis, the final output will remain unaltered.\par 
Now, it is trivial to argue that any causally separable combination of these two actions, i.e., $\mathcal{S}_{\lambda}(\mathcal{N}_{PF},\mathcal{N}_{GAD})=\lambda \mathcal{N}_{PF}\circ\mathcal{N}_{GAD}+(1-\lambda)\mathcal{N}_{GAD}\circ\mathcal{N}_{PF}$ will produce nothing other than $\tau_{p}$ for any arbitrary input state, where $\mathcal{S}_{\lambda}(\mathcal{N}_{PF},\mathcal{N}_{GAD})$ is a supermap acting on the super operator space of $\mathcal{L}(\mathcal{H})$ and preserves the CPTP-ness of super operators. Hence, the free energy of the initial state can not increase under these channel actions.
\end{proof}
\begin{figure}[!htb]
	\includegraphics[scale=0.35]{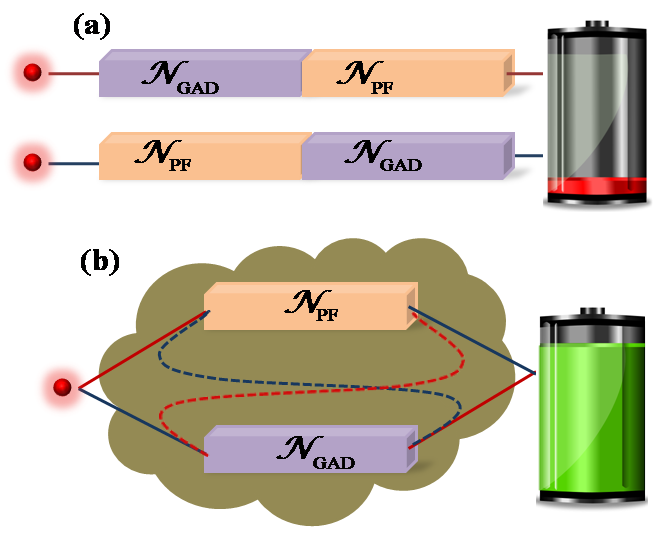}
	\caption{(a) No activation of free energy in the final state (depicted as a battery) for the definite causal order of occurrence of the thermal channels $\mathcal{N}_{GAD}$ and $\mathcal{N}_{PF}$. (b) However the same is activated when these channels appear in an indefinite causal order.}\label{proc}
\end{figure}
It is important to mention that, in the above theorem we have considered free energy of a quantum state with respect to a fixed bath, which is at a temperature corresponding to $\tau_{p}$ for $\frac{1}{2}\leq p\leq 1$. Assuming the governing system Hamiltonian as $H=\ketbra{1}{1}$, it is possible to assign a temperature $T_{p}=\frac{1}{k_{B} \ln(\frac{p}{1-p})}$ corresponding to $\tau_{p}$ as, where $k_{B}$ is the Boltzmann constant. So the free energy in this case can be defined as
\begin{equation}\label{e2}
	F_{p}(\rho)=\Tr(\rho H)-\frac{1}{k_{B}\ln(\frac{p}{1-p})}\times\ln 2\times S(\rho)
\end{equation}
where, $S(\rho)$ is the von-Neumann entropy corresponding to the initial quantum state $\rho$. It is easy to verify that the free energy of $\rho$, with respect to this fixed bath will be minimum when the state is identical to the bath state itself. Hence, in our case $\tau_{p}$ is the minimum free energetic state corresponding to the bath fixed at temperature $T_{p}$.\par 
Recently, in \cite{adesso'PRL} the authors have introduced a new framework, where the work extraction from the local marginal of a bipartite correlation is considered with the local operation and classical communication (LOCC) assistance from the other party. More precisely, to extract work locally from a correlated state, one of the constituents is allowed to apply only thermal operations on her/his part, whereas the other is allowed to apply any CPTP map and to communicate classically. In a similar fashion we design a framework, where along with the system qubit one party (say, Alice) is allowed to access an ancillary system and to apply all possible energy preserving global unitaries on the system-ancilla joint state. Thereafter she transmits the system particle through the channels $\mathcal{N}_{GAD}$ and $\mathcal{N}_{PF}$ successively to Bob. Then performing a measurement on the ancillary state and communicating the outcome, Alice can help Bob to extract work from the transmitted qubit in his possession. However, the action of any incoherent combination of these two channels produces a fixed quantum state with complete ignorance on the input as well as the correlation shared with its ancillary counterpart. Hence, even by invoking the assistance from ancillary particle the free energy of the system with respect to the fixed bath can not be enhanced, thereby in agreement with Theorem 1.
\subsection{Thermal enhancement of system state in expense of a resourceful state}
Instead of using the ancillary quantum state in the above mentioned framework, we can use it to control the order of occurrence of the thermal channels. More precisely, when the ancillary qubit is in $\ket{0}$ state then the generalized amplitude damping channel $\mathcal{N}_{GAD}$ will act after the application of phase flip channel $\mathcal{N}_{PF}$, and the channels will be in reverse order when the ancillary qubit is in $\ket{1}$ state. This controlling process will act as a super map on these two channels, and hence produce $\mathcal{S}_{\phi}(\mathcal{N}_{PF}, \mathcal{N}_{GAD})$, where $\phi$ denotes the state of controller qubit $\ket{\phi}$. Although the input channels are same the action of super map $\mathcal{S}_{\phi}$ is not similar to that as mentioned in Theorem 1, because the controlling qubit $\ket{\phi}$ is allowed to possess superposition in the controller basis, i.e., $\{\ket{0},\ket{1}\}$. Hence the action of this super map on the system-controller joint state in terms of there Kraus representations can be visualized as,
\begin{equation}\label{e3}
	\mathcal{S}_{\phi}(\mathcal{N}_{PF}, \mathcal{N}_{GAD})(\rho\otimes\ketbra{\phi}{\phi})= \sum_{i=0}^{3}\sum_{j=0}^{1}K_{ij}(\rho\otimes\ketbra{\phi}{\phi})K^{\dagger}_{ij} 
\end{equation} 
The Kraus operators $(K_{ij})$ are given by,
\begin{equation}\label{e4}
	K_{ij}=E_{i}F_{j}\otimes\ketbra{0}{0}+F_{j}E_{i}\otimes\ketbra{1}{1}
\end{equation} 
where, $E_{i}$'s and $F_{j}$'s are same as defined earlier.\\
In the present work, we use the controller in the state $\ket{\phi} =\frac{\ket{0}+\ket{1}}{\sqrt{2}}$, which is obviously a resourceful state with respect to the thermal bath kept at a fixed temperature $T_{p}$. Notably, using Eq.\eqref{e2} one can calculate the amount of free energy stored in the controller state which yields the value $0.5$. Without loss of generality we choose $p=q$ for $E_{i}$ and $F_{j}$'s (a more general study for these channels is discussed in the {\it Appendix}). Further, by taking $\rho=r\ketbra{0}{0}+(1-r)\ketbra{1}{1}$ as the initial system qubit, the final bipartite state under the action of this causally inseparable combination will take the form, 
\begin{widetext}
	\begin{align}\label{e5}
	\nonumber\mathcal{S}_{\phi}(\mathcal{N}_{PF}, \mathcal{N}_{GAD})(\rho\otimes\ketbra{+}{+})=\lambda\rho_{-}\otimes\ketbra{-}{-}+(1-\lambda)\rho_{+}\otimes\ketbra{+}{+} &\\
	\nonumber \text{where, } \lambda=(1-p)(p+r-2pr), &\\
	\rho_{\pm}=p_{\pm}\ketbra{0}{0}+(1-p_{\pm})\ketbra{1}{1},~~~  p_{-}=\frac{p(1-p)(1-r)}{\lambda}~~\text{and}~~   p_{+}=\frac{p(p+r-pr)}{(1-\lambda)}.
		\end{align}
\end{widetext}
Now, by making a measurement on the controller qubit and communicating the result, one can prepare $\rho_{\pm}$ with probabilities $(1-\lambda)$ and $\lambda$ respectively to the channel receiver.\par 
{\it Case 1. fixed bath:} In question of work extraction from the conditional output state, first we will consider the scenario where the receiver has an access to a thermal bath fixed at a temperature $T_{p}$. As a result, the average amount of work potential stored in the receiver's possession is given by
 \begin{equation}\label{e6}
 	W_{avg}=\lambda\times(F(\rho_{-})-F(\tau_{p}))+(1-\lambda)\times(F(\rho_{+})-F(\tau_{p}))
 \end{equation} 
 It is important to mention that for $r=0.5$, i.e., when the input state is maximally mixed, then the final state will be pinned to $\tau_{p}$. However, for $0.5<r\leq1$ the final state is different and as a result it possesses a certain amount of work potential in terms of free energy difference with the bath state $\tau_{p}$. Referring to Fig\ref{wrf} it is clear that for the range $0.5<r\leq1$, the difference between the two curves will quantify the amount of extractable work in case of causally inseparable combinations of two thermal channels $\mathcal{N}_{PF}$ and $\mathcal{N}_{GAD}$.\\
 \begin{figure}[!htb]
 	\includegraphics[scale=0.5]{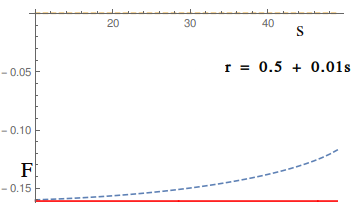}
 	\caption{Variation of final free energy with the input state parameter is depicted for channel parameters $p=q=0.8$. The \textcolor{blue}{dashed} curve denotes the free energy of the final output under the causally inseparable combinations of the channels $\mathcal{N}_{PF}$ and $\mathcal{N}_{GAD}$, whereas the \textcolor{red}{thick} line stands for the final free energy under causally separable combinations of these two channels.}\label{wrf}
 \end{figure}
  
  {\it Case 2. varying bath:} Further, it is interesting to study a different framework considering the reference bath in Bob's possession is not fixed but same as the system qubit temperature. Although the final output state is pinned at $\tau_{p}$ irrespective of the input qubit, even under causally inseparable framework, the extractable work potential in the final output state is nonzero and varies with the temperature associated to the system state. It is obvious that only for the choice of an input state as $\tau_{p}$, the free energy of the final state is minimum with respect to the thermal bath of input temperature, hence no work can be extracted. However, for any other choice of the input state, the final state possess a nonzero work potential with respect to the bath same as of input temperature. Interestingly, coherent control on the order of occurrence of the thermal channels $\mathcal{N}_{PF}$ and $\mathcal{N}_{GAD}$ can enhance the potential work storage of the final output. This enhancement with respect to the ground state probability of the input thermal state is depicted in Fig \ref{wrv}.
\subsection{Thermal enhancement of system state without expense of resourceful state}
The results presented above are weaker in the sense that they exploit a pure state with free energy value $0.5$ in energy unit, which is a resourceful state from the perspective thermodynamic work extraction. Recently, there has been a critique on the use of a perfect channel for the controlling qubit \cite{brukner'PRA} for the works pertaining to  information theoretic causal enhancements \cite{Ebler'PRL,our'arxiv}. However, it is also important to mention that these results have sufficient relevance to introduce a completely new framework in the context of information theory. Interestingly, the present work is a potential candidate in the resource theoretic direction of these enhancements. More precisely, in this subsection we will show that the thermodynamic advancement of the input thermal state is possible by controlling the action of thermal maps coherently with the help of a completely free state, i.e., the same thermal state. To appreciate the strength of this result we will first compare two different resource theoretic scenarios in Table \ref{tab1}.
\begin{center}
	\begin{table}[!htb]
			\begin{tabular}{ |c|c|c| } 
			\hline
			Resource Theory & Free State & Free Operation \\ 
			\hline\hline
			Purity & $\frac{\mathbb{I}}{2}$ & Noisy operation \\
			\hline 
			Thermodynamics& Bath state $\tau_{\beta}$ & Thermal operation \\ 
			\hline
		\end{tabular}
	\caption{Resource theory of purity is concerned with classical information processing tasks, whereas thermodynamics quantifies the amount of extractable work in presence of a thermal bath $\tau_{\beta}$.}\label{tab1}
	\end{table}
\end{center}
In the direction of classical information processing tasks, the authors in \cite{Ebler'PRL} exploit the causally inseparable order of occurrence of two identical completely noisy channels and succeed to convey a significant amount of classical information. However, the controlling qubit is a pure state, containing highest amount of resource value in context of purity. Moreover, it is impossible to obtain a partial advancement using the completely free state of the corresponding resource theory, i.e, $\frac{\mathbb{I}}{2}$ as a controller state. Similarly, from the perspective of thermodynamic resource theory, the resource value assigned to a quantum state depends upon the amount of free energy stored in it with respect to a given thermal bath at inverse temperature $\beta$. It is easy to show that the corresponding bath state contains minimum free energy with respect to itself, hence serves as a free state in the particular resource theoretic paradigm. However, in contrast to the information processing tasks, this free state is useful to control the sequence of occurrence of thermal channels and to enhance the thermodynamic potential of the input state, even if it is a free state. It is also important to mention that indefinite causal order of occurrence can not activate the quantum capacity for the channels with zero classical capacity \cite{our'arxiv}.\par 
   \begin{figure}[!b]
	\includegraphics[scale=0.5]{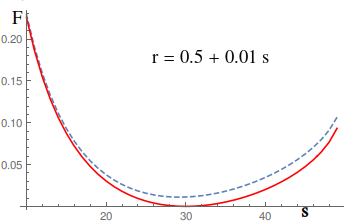}
	\caption{The free energy difference between the final and the input state with respect to the input bath is plotted for $p=q=0.8$. The potential work content of the final output in causally inseparable and separable combinations of two thermal channels are denoted by the \textcolor{blue}{dashed} and \textcolor{red}{thick} curves respectively.}\label{wrv}
\end{figure}
In this scenario we will consider a slightly different controlling sequence: if the controller state is $\ket{+}$ then the generalized amplitude damping channel will act after the action of phase flip channel, and in the reverse order when the controller state is $\ket{-}$. As a result, Eq.\eqref{e5} takes the form
\begin{widetext}
		\begin{align}\label{e7}
	\nonumber\mathcal{S}_{\phi}(\mathcal{N}_{PF}, \mathcal{N}_{GAD})(\rho^{\otimes 2})= \lambda\rho_{0}\otimes\ketbra{0}{0}+(1-\lambda)\rho_{1}\otimes\ketbra{1}{1} &\\
	\nonumber \text{where, } \lambda=(p+2r(1-p))(1-r+p(2r-1)), &\\
	\rho_{k}=p_{k}\ketbra{0}{0}+(1-p_{k})\ketbra{1}{1}, k\in\{0,1\},~~~  p_{0}=\frac{p[(1-p)(1+2 r^{2})+(3p-2) r]}{\lambda}~~\text{and}~~   p_{1}=\frac{p(1-r)[p+2r(1-p)]}{(1-\lambda)}.
	\end{align}

 Under this circumstance we have again considered two different scenarios (with both the fixed as well as varying temperature bath) and the final free energy is depicted in the Fig\ref{wor}.\par 
 \begin{figure}[!htb]
	\includegraphics[scale=0.5]{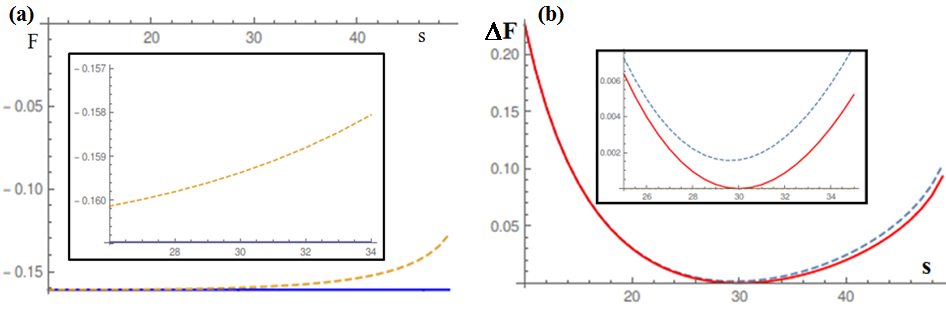}
	\caption{Channel parameter is chosen as $p=q=0.8$ and input state probability $r=0.5 + 0.01 s$ as mentioned earlier. (a) {\it Fixed bath scenario:} The \textcolor{blue}{thick} straight line at the bottom denotes the free energy of the final output under the causally separable combination of the channels $\mathcal{N}_{PF}$ and $\mathcal{N}_{GAD}$, whereas the \textcolor{brown}{dashed} line stands for the final free energy under causally inseparable combination of these two channels using the free thermal state as the controlling qubit. (b) {\it Varying bath scenario:} The \textcolor{red}{thick} curve denotes the free energy difference between the final output (under the causally separable combination of the channels) and the initial state. The \textcolor{blue}{dashed} line stands for the same under causally inseparable combination of these two channels with thermal controller.}\label{wor}
\end{figure}
\end{widetext}
 \section{Conclusion}
Here we have shown that thermal channels can be made useful to enhance the extractable amount of work in the final output, under the coherent control of occurrence. In particular, the combination of two specific thermal channels, namely {\it phase-flip} and {\it generalized amplitude damping} channel in any definite order of occurrence maps every input state to a fixed thermal state. Hence the extractable amount of work from the channel output with respect to the bath of that same temperature is zero independent of the input states. Interestingly, we have shown that if these two channels occur in causally inseparable order then the final free energy can be enhanced with respect to the same thermal bath, leading to an activation of the extractable work. However, this enhancement can not surpass the value of the initial free energy stored in the quantum state, thereby respecting the second law of thermodynamics. Further, our result has important consequence from the resource theoretic viewpoint also, in the sense that, one can activate these thermal channels even using the resource free thermal state as a controller. However, it associates a measurement cost in order to be consistent with the second law. It has been shown in \cite{abott} that advancements through the causally inseparable combination of unital channels is achievable also under the {\it definite} causal structure, but controlling their path coherently. However, the situation is much more complicated in case of a non-unital qubit channel and there is no deterministic way to construct the global channel mapping for the coherent control of the paths. In this sense, our result has a significant aspect to establish the causal advancement through the application of non-unital channel, i.e., the {\it generalized amplitude damping} channel. Hence, it may be interesting to study further if this causal inseparability of channel actions can be visualized as a coherent control of paths in definite causal order.
 \section{Acknowledgment}
 M.A. would like to acknowledge the CSIR project 09/093(0170)/2016-EMR-I for financial support.
 
\begin{widetext}
    \section{Appendix}
 	In this article our main goal is to establish the causal enhancement in the context of thermal channels, which is exhibited in the main text for specific parametric values of two thermal channels and considering all possible thermal states as input. Here, we will complete our study considering the full parametric range of the {\it Generalized Amplitude Damping} channel. There is no change in free energy for the variation of channel parameter $q$ for  $\mathcal{N}_{PF}(q)$ in causally separable order of occurrence, as the appearance of the pin map $\mathcal{N}_{GAD}(p)$ will fix the final output (for a fixed $p$) independent of the {\it phase-flip parameter} $q$. However, in the causally inseparable order of occurrence of these two channels, the final free energy is sensitive to the phase flip parameter $q$. In the following the variation of final free energy with respect to the chosen bath is depicted with the variation in amplitude damping parameter $p$ and input state parameter $s$, where $r=0.5+0.01s$ and $r$ is the population in the ground state of the input qubit (see Fig \ref{fwr}).\par
 	\begin{figure}[!htb]
 		\includegraphics[scale=0.5]{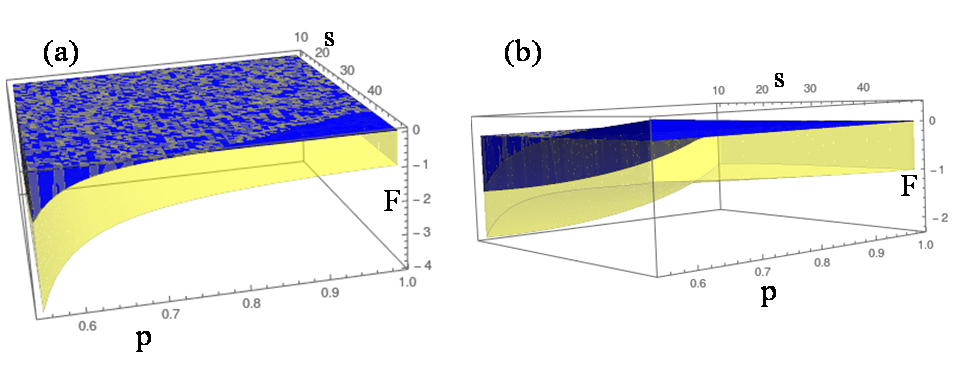}
 		\caption{\textbf{Thermodynamic advancement using resource in the controller qubit.} For both the plots  \textcolor{yellow}{yellow} region denotes the variation of final free energy under causally separable order of occurrence for $\mathcal{N}_{GAD}$ and $\mathcal{N}_{PF}$ and the  \textcolor{blue}{blue} region stands for their inseparable causal combination. The phase flip parameter $q$ is chosen as $0.3$. In (a) the final free energy is calculated with respect to the bath fixed by the $\mathcal{N}_{GAD}$ channel parameter $p$, wherein (b) the same is calculated by assuming that the receiver has an access to the same bath as that of the input qubit.}\label{fwr}
 	\end{figure}
 	   Further, in Fig \ref{fwor} the same variation is studied without any thermodynamic resource in the controller qubit. In this case, it is evident from the plots that the difference in final free energy is not significant for the complete range of channel and state parameters because the thermal state in the controller qubit possesses lesser coherence than the resourceful state $\ket{+}$ (as in Fig \ref{fwr}) to bring these two channels in causally indefinite order of occurrence. On the other hand, this can be explained from the perspective of thermodynamic resource theory by noting that thermal states are absolutely free with respect to their thermal bath. Hence, lesser resource in the input can not arbitrarily enhance the work potential due to the laws of thermodynamics.
 	   \begin{figure}[!htb]
 	   	\includegraphics[scale=0.5]{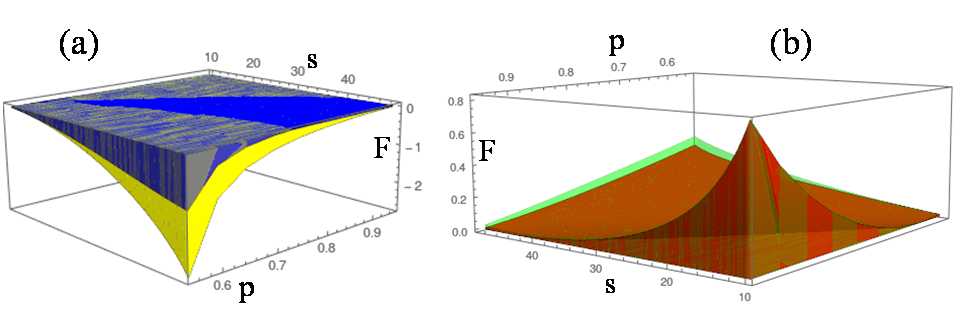}
 	   	\caption{\textbf{Thermodynamic advancement using resource free controller qubit.} The phase flip parameter $q$ is chosen as earlier. (a) The variation of the final free energy for causally definite sequence of these two channels is depicted by the \textcolor{yellow}{yellow} region, whereas the \textcolor{blue}{blue} region stands for causally inseparable order of occurrence. (b) Here the \textcolor{green}{green} region depicts the causally inseparable case and the \textcolor{red}{red} region depicts causally separable case.}\label{fwor}
 	   \end{figure}
 	\end{widetext}
              
\end{document}